\documentclass[a4paper]{article}
\usepackage{graphicx}

\usepackage{onecolceurws}

\usepackage{isabelle,isabellesym}
\isabellestyle{literal}

\usepackage{amsthm}
\usepackage{amsmath,amsfonts}

{\itshape}{\rmfamily}
{\itshape}{\rmfamily}
{\bfseries}{\itshape}
{\bfseries}{\itshape}
{\itshape}{\rmfamily}
{\itshape}{\rmfamily}
\newtheorem{lemma}{Lemma}{\bfseries}{\itshape}
{\itshape}{\rmfamily}
{\itshape}{\rmfamily}
{\itshape}{\rmfamily}
{\bfseries}{\itshape}
{\itshape}{\rmfamily}
{\itshape}{\rmfamily}
{\itshape}{\rmfamily}

\usepackage[most]{tcolorbox}

\newtcolorbox{isaframe}[1][]
   { blanker, 
    left=3mm, right=3mm, top=1mm, bottom=1mm,
     borderline west={1pt}{0pt}{blue},
     before upper=\setlength{\parindent}{0pt},
     parbox=true, #1}

\usepackage[draft=false,breaklinks,colorlinks]{hyperref}

\usepackage[nameinlink]{cleveref}
\crefname{theorem}{theorem}{theorems}
\crefname{corollary}{corollary}{corollaries}
\crefname{example}{example}{examples}
\crefname{lemma}{lemma}{lemmas}
\crefname{proposition}{proposition}{propositions}
\crefname{definition}{definition}{definitions}
\crefname{observation}{observation}{observations}

\newcommand{\isaterm}[1]{\texttt{{#1}}}

\newcommand\rev[1]{%
\ensuremath{\mathop\mathrm{rev}%
\ifx&#1&%
\else
  {\left( #1 \right)}
\fi
}}

\DeclareMathOperator{\id}{id}
\newcommand{\alist}{\isaterm{{\isacharprime}a\ list}}
\newcommand{\blist}{\isaterm{{\isacharprime}b\ list}}
\newcommand{\alistlist}{\isaterm{{\isacharprime}a\ list list}}

\newcommand{\Nil}{\isaterm{Nil}}
\newcommand{\Cons}{\isaterm{Cons}}

\title{Producing symmetrical facts for lists induced by the list reversal mapping in Isabelle/HOL}

\author{
Martin Raška \\ Charles University \\ Czech Republic
\and
Štěpán Starosta \\ Czech Technical University in Prague \\ Czech Republic \\ stepan.starosta@fit.cvut.cz
}

\institution{}

\begin{document}
\maketitle

\begin{abstract}
Many facts possess symmetrical counterparts that often require a separate formal proof, depending on the nature of the involved symmetry.
We introduce a method in Isabelle/HOL which produces such a symmetrical fact for the list datatype and the symmetry induced by the list reversal mapping.
The method is implemented as an attribute and its result is based on user-declared symmetry rules.
Besides general rules, we provide rules that are aimed to be applied in the domain of Combinatorics on Words.
\end{abstract}
\vskip 32pt

\section{Introduction}

While formalizing a piece of mathematical knowledge, one probably hopes that some part of the tedious work will be done by the machine. 
One such mechanical tasks are proofs that follow ``by symmetry'' which can be seen as a variation of ``without loss of generality'' \cite{wlog}.
One such ``by symmetry'' usually stands for a proper description of the symmetry involved and the procedure of how lemmas involving the symmetry should be used to obtain the symmetrical claim.

In this article, we exhibit a partial, yet quite useful, solution to ``by symmetry'' in the case of lists and the reversal mapping in the proof assistant Isabelle/HOL \cite{IsabelleHOLBook}.
The reversal, or mirror mapping, is the mapping reversing the order of elements in a list.
This mapping interconnects many pairs of definitions over lists in the spirit of the following duality: the list $p$ is a prefix of the list $w$ if and only if the reversal of $p$ is a suffix of the reversal of $w$.
We situate this solution in the context of Combinatorics on words, a mathematical domain which studies words, i.e., lists and their various properties including equations on words.

First, we give a short overview of mathematical context along with examples of the symmetry in question.
In \Cref{sec:solution}, we shortly describe possible approaches to the solution and then we describe our solution which is part of the ongoing project of formalization of Combinatorics on Words \cite{CoW_gitlab}.
We conclude by describing the limits of the current solution in \Cref{sec:limits} and conclude by final remarks in \Cref{sec:conclusion}

\section{Mathematical context and examples of the symmetry} \label{sec:context_and_examples}

We work with \emph{words}, which are finite sequences $(a_i)_{i=0}^{n}$ with $a_i \in A$ with $A$ usually being a finite set.
The set of all words over $A$ is denoted $A^*$ (where $^*$ is the Kleene star).
The \emph{reversal mapping}, denoted $\rev{}$, is a mapping $A^* \to A^*$ which maps the word $w = (a_i)_{i=0}^{n}$ to the word $\rev{w} = (a_{n-i})_{i=0}^{n}$, or simply put, it reads the letters of the word in the reverse order.
The reversal mapping is an \emph{involutive antimorphism} with respect to the operation of concatenation of two words, that is, $\rev{} \circ \rev{} = \id$ and $\rev{v \cdot w} = \rev{w} \cdot \rev{v}$ where $\cdot$ is the binary operation of concatenation.
It follows that \rev{} is also a bijection.

In our ongoing project \cite{CoW_gitlab} of formalization of Combinatorics on Words, we formalize many elementary preparatory lemmas dealing with a handful of notions.
Many of these notions have a symmetrical counterpart, and many facts are symmetrical, and their proof is just copy and paste of the proof of the original lemma.
We continue with examples that exhibit this symmetry.

\subsection{Example 1} \label{sec:ex1}

If a word $w$ can be written as a concatenation of two words $p$ and $s$, i.e., $w = p \cdot s$, we say that $p$ is a \emph{prefix} of $w$ and $s$ is a \emph{suffix}.
An elementary example of a symmetrical pair of claims involving prefix and suffix is the following.

\begin{lemma} \label{pair1}
If $p$ is a prefix of $v$, then $p$ a prefix of $v \cdot w$.
\end{lemma}

\begin{proof}
If $p$ is a prefix of $v$, there exists a word $s$ such that $v = p \cdot s$.
Hence, $v\cdot w = p\cdot s\cdot w$, and $p$ is a prefix of $v \cdot w$.
\end{proof}

By the symmetrical counterpart of \Cref{pair1} we mean the following claim.

\begin{lemma}\label{pair2}
If $s$ is suffix of $v$, then $s$ is a suffix of $w \cdot v$.
\end{lemma}

Its proof can be done as the presented proof of \Cref{pair1}, however, stating in follows ``by symmetry'' from \Cref{pair1} would be no exception in literature.

In order to formally exploit the symmetry in a full proof, we have to make the intended symmetry between a prefix and a suffix explicit:
\begin{lemma}\label{p_to_s}
The word $p$ is a prefix of $w$ if and only if $\rev{p}$ is a suffix of $\rev{w}$.
\end{lemma}

A full proof of \Cref{pair2}, by symmetry, is as follows.

\begin{proof}[Proof of \Cref{pair2}]
Fix $w$ and assume that $s$ is a suffix of $v$.
Let $s'$, $v'$, and $w'$ be the words such that
\[
s' = \rev{s}, \quad v' = \rev{v} \quad \text { and } \quad w' = \rev{w}.
\]
As the reversal is an involution, it follows that
\[
s = \rev{s'}, \quad v = \rev{v'} \quad \text { and } \quad w = \rev{w'}.
\]
As $s$ is a suffix of $v$, the word $\rev{s'}$ is a suffix of $\rev{v'}$.
By \Cref{p_to_s}, $s'$ is a prefix of $v'$.
Using \Cref{pair1}, $s'$ is a prefix of $v' \cdot w'$.
Again, by \Cref{p_to_s}, $\rev{s'}$ is a suffix of $\rev{v' \cdot w'}$.
Since
\[
\rev{v' \cdot w'} = \rev{w'}\cdot \rev{v'} = w\cdot v,
\]
we conclude that $s$ is a suffix of $w\cdot v$.
\end{proof}

\subsection{Example 2} \label{sec:ex2}

Since the next examples are in the framework of Isabelle/HOL, we first recall our setting.
A word is represented by the datatype of list, which is is specified via 2 constructors: \Nil{} (denoted \isaterm{{\isacharbrackleft}{\kern0pt}{\isacharbrackright}{\kern0pt}}), the empty list/word, and \Cons{} (denoted \isaterm{{\isacharhash}}), the recursive constructor allowing to add an element to the list at its beginning.
The reversal mapping is represented by the function \isaterm{rev}:

\begin{isaframe}
\isacommand{primrec}\isamarkupfalse%
\ rev\ {\isacharcolon}{\kern0pt}{\isacharcolon}{\kern0pt}\ {\isachardoublequoteopen}{\isacharprime}{\kern0pt}a\ list\ {\isasymRightarrow}\ {\isacharprime}{\kern0pt}a\ list{\isachardoublequoteclose}\ \isakeyword{where}\isanewline
{\isachardoublequoteopen}rev\ {\isacharbrackleft}{\kern0pt}{\isacharbrackright}{\kern0pt}\ {\isacharequal}{\kern0pt}\ {\isacharbrackleft}{\kern0pt}{\isacharbrackright}{\kern0pt}{\isachardoublequoteclose}\ {\isacharbar}{\kern0pt}\isanewline
{\isachardoublequoteopen}rev\ {\isacharparenleft}{\kern0pt}x\ {\isacharhash}{\kern0pt}\ xs{\isacharparenright}{\kern0pt}\ {\isacharequal}{\kern0pt}\ rev\ xs\ {\isacharat}{\kern0pt}\ {\isacharbrackleft}{\kern0pt}x{\isacharbrackright}{\kern0pt}{\isachardoublequoteclose}%
\end{isaframe}

with \isaterm{{\isacharat}} being the notation for list \isaterm{append}, i.e., concatenation of two words.

The predicates for prefix and suffix are already part of the Isabelle distribution in the theory HOL-Library.Sublist:

\begin{isaframe}
\isacommand{definition}\isamarkupfalse%
\ prefix\ {\isacharcolon}{\kern0pt}{\isacharcolon}{\kern0pt}\ {\isachardoublequoteopen}{\isacharprime}{\kern0pt}a\ list\ {\isasymRightarrow}\ {\isacharprime}{\kern0pt}a\ list\ {\isasymRightarrow}\ bool{\isachardoublequoteclose} 
\ \ \isakeyword{where}\ {\isachardoublequoteopen}prefix\ xs\ ys\ {\isasymlongleftrightarrow}\ {\isacharparenleft}{\kern0pt}{\isasymexists}zs{\isachardot}{\kern0pt}\ ys\ {\isacharequal}{\kern0pt}\ xs\ {\isacharat}{\kern0pt}\ zs{\isacharparenright}{\kern0pt}{\isachardoublequoteclose}\isanewline
\isanewline
\isacommand{definition}\isamarkupfalse%
\ suffix\ {\isacharcolon}{\kern0pt}{\isacharcolon}{\kern0pt}\ {\isachardoublequoteopen}{\isacharprime}{\kern0pt}a\ list\ {\isasymRightarrow}\ {\isacharprime}{\kern0pt}a\ list\ {\isasymRightarrow}\ bool{\isachardoublequoteclose} 
\ \ \isakeyword{where}\ {\isachardoublequoteopen}suffix\ xs\ ys\ {\isacharequal}{\kern0pt}\ {\isacharparenleft}{\kern0pt}{\isasymexists}zs{\isachardot}{\kern0pt}\ ys\ {\isacharequal}{\kern0pt}\ zs\ {\isacharat}{\kern0pt}\ xs{\isacharparenright}{\kern0pt}{\isachardoublequoteclose}
\end{isaframe}




The second example is constituted by the pair of symmetric definitions of the first and the last letter of a word, i.e., element of a list.
In Isabelle/HOL, the first element of a list is its head, realized as one of two selectors, named \isaterm{hd}, of the list constructor \Cons{}.
The last letter is the recursive function \isaterm{last}:

\begin{isaframe}
\isacommand{primrec}\isamarkupfalse%
\ {\isacharparenleft}{\kern0pt}nonexhaustive{\isacharparenright}{\kern0pt}\ last\ {\isacharcolon}{\kern0pt}{\isacharcolon}{\kern0pt}\ {\isachardoublequoteopen}{\isacharprime}{\kern0pt}a\ list\ {\isasymRightarrow}\ {\isacharprime}{\kern0pt}a{\isachardoublequoteclose}\ \isakeyword{where}\isanewline
\ \ {\isachardoublequoteopen}last\ {\isacharparenleft}{\kern0pt}x\ {\isacharhash}{\kern0pt}\ xs{\isacharparenright}{\kern0pt}\ {\isacharequal}{\kern0pt}\ {\isacharparenleft}{\kern0pt}if\ xs\ {\isacharequal}{\kern0pt}\ {\isacharbrackleft}{\kern0pt}{\isacharbrackright}{\kern0pt}\ then\ x\ else\ last\ xs{\isacharparenright}{\kern0pt}{\isachardoublequoteclose}%
\end{isaframe}

\noindent given in the main theory List.
To obtain a simple enough symmetry rule for {\tt hd} and {\tt last}, it suffices to notice
that they behave the same way on the empty list.

\begin{isaframe}
\isacommand{lemma}
\ hd{\isacharunderscore}{\kern0pt}last{\isacharunderscore}{\kern0pt}Nil{\isacharcolon}{\kern0pt}\ {\isachardoublequoteopen}hd\ {\isacharbrackleft}{\kern0pt}{\isacharbrackright}{\kern0pt}\ {\isacharequal}{\kern0pt}\ last\ {\isacharbrackleft}{\kern0pt}{\isacharbrackright}{\kern0pt}{\isachardoublequoteclose} 
\isacommand{unfolding}\isamarkupfalse%
\ hd{\isacharunderscore}{\kern0pt}def\ last{\isacharunderscore}{\kern0pt}def\ \isacommand{by}\isamarkupfalse%
\ simp%
\isanewline
\isanewline
\isacommand{lemma}\isamarkupfalse%
\ hd{\isacharunderscore}{\kern0pt}rev{\isacharunderscore}{\kern0pt}last{\isacharcolon}{\kern0pt}\ {\isachardoublequoteopen}hd{\isacharparenleft}{\kern0pt}rev\ xs{\isacharparenright}{\kern0pt}\ {\isacharequal}{\kern0pt}\ last\ xs{\isachardoublequoteclose} 
\isacommand{by}\isamarkupfalse%
\ {\isacharparenleft}{\kern0pt}induct\ xs{\isacharcomma}{\kern0pt}\ simp\ add{\isacharcolon}{\kern0pt}\ hd{\isacharunderscore}{\kern0pt}last{\isacharunderscore}{\kern0pt}Nil{\isacharcomma}{\kern0pt}\ simp{\isacharparenright}{\kern0pt}%
\end{isaframe}

The pair of symmetrical claims is the following.

\begin{isaframe}
\isacommand{lemma}\isamarkupfalse%
\ example{\isadigit{2}}{\isacharcolon}{\kern0pt}\ {\isachardoublequoteopen}u\ {\isasymnoteq}\ {\isacharbrackleft}{\kern0pt}{\isacharbrackright}{\kern0pt}\ {\isasymLongrightarrow}\ \ prefix\ u\ v\ {\isasymLongrightarrow}\ hd\ u\ {\isacharequal}{\kern0pt}\ hd\ v{\isachardoublequoteclose}
\isanewline
\isanewline
\isacommand{lemma}\isamarkupfalse%
\ example2\_sym{\kern0pt}{\isacharcolon}{\kern0pt}\ {\isachardoublequoteopen}u\ {\isasymnoteq}\ {\isacharbrackleft}{\kern0pt}{\isacharbrackright}{\kern0pt}\ {\isasymLongrightarrow}\ \ suffix\ u\ v\ {\isasymLongrightarrow}\ last\ u\ {\isacharequal}{\kern0pt}\ last\ v{\isachardoublequoteclose}
\end{isaframe}

The goal is to obtain example2\_sym from example{\isadigit{2}} by symmetry.
We proceed analogously to the proof of \Cref{pair2} above using standard methods in Isabelle/HOL:

\begin{isaframe}
example{\isadigit{2}}{\isacharbrackleft}{\kern0pt}of\ {\isachardoublequoteopen}rev\ u{\isachardoublequoteclose}\ {\isachardoublequoteopen}rev\ v{\isachardoublequoteclose}{\isacharcomma}{\kern0pt}\ unfolded\ rev{\isacharunderscore}{\kern0pt}is{\isacharunderscore}{\kern0pt}Nil{\isacharunderscore}{\kern0pt}conv\ suffix{\isacharunderscore}{\kern0pt}to{\isacharunderscore}{\kern0pt}prefix{\isacharbrackleft}{\kern0pt}symmetric{\isacharbrackright}{\kern0pt}\ hd{\isacharunderscore}{\kern0pt}rev{\isacharunderscore}{\kern0pt}last{\isacharbrackright}{\kern0pt}
\end{isaframe}
where \isaterm{rev{\isacharunderscore}{\kern0pt}is{\isacharunderscore}{\kern0pt}Nil{\isacharunderscore}{\kern0pt}conv}
is
\isaterm{{\isacharparenleft}{\kern0pt}rev\ xs\ {\isacharequal}{\kern0pt}\ {\isacharbrackleft}{\kern0pt}{\isacharbrackright}{\kern0pt}{\isacharparenright}{\kern0pt}\ {\isacharequal}{\kern0pt}\ {\isacharparenleft}{\kern0pt}xs\ {\isacharequal}{\kern0pt}\ {\isacharbrackleft}{\kern0pt}{\isacharbrackright}{\kern0pt}{\isacharparenright}{\kern0pt}}
and \isaterm{suffix{\isacharunderscore}{\kern0pt}to{\isacharunderscore}{\kern0pt}prefix{\isacharbrackleft}{\kern0pt}symmetric{\isacharbrackright}{\kern0pt}} is \isaterm{prefix\ {\isacharparenleft}{\kern0pt}rev\ xs{\isacharparenright}{\kern0pt}\ {\isacharparenleft}{\kern0pt}rev\ ys{\isacharparenright}{\kern0pt}\ {\isacharequal}{\kern0pt}\ suffix\ xs\ ys}.
That is, we instantiate every variable of example{\isadigit{2}} by its reversal to obtain
\begin{isaframe}
{\isachardoublequoteopen}rev\ u\ {\isasymnoteq}\ {\isacharbrackleft}{\kern0pt}{\isacharbrackright}{\kern0pt}\ {\isasymLongrightarrow}\ prefix\ {\isacharparenleft}{\kern0pt}rev\ u{\isacharparenright}{\kern0pt}\ {\isacharparenleft}{\kern0pt}rev\ v{\isacharparenright}{\kern0pt}\ {\isasymLongrightarrow}\ hd\ {\isacharparenleft}{\kern0pt}rev\ u{\isacharparenright}{\kern0pt}\ {\isacharequal}{\kern0pt}\ hd\ {\isacharparenleft}{\kern0pt}rev\ v{\isacharparenright}{\kern0pt}{\isachardoublequoteclose},
\end{isaframe}

and then rewrite the terms using appropriate symmetry rules, 
via the unfolded attribute (which is analogous to what was done in the proof of \Cref{pair2} above).
We end up with example2\_sym and the proof by symmetry is done.

\subsection{Example 3} \label{sec:ex3}

The next example is the following pair of symmetric facts.

\begin{isaframe}
\isacommand{lemma}\isamarkupfalse%
\ example{\isadigit{3}}{\isacharcolon}{\kern0pt}\ {\isachardoublequoteopen}prefix\ u\ {\isacharparenleft}{\kern0pt}p\ {\isacharat}{\kern0pt}\ w\ {\isacharat}{\kern0pt}\ q{\isacharparenright}{\kern0pt}\ {\isasymLongrightarrow}  \ length\ p\ {\isasymle}\ length\ u\ \ {\isasymLongrightarrow}\ length\ u\ {\isasymle}\ length\ {\isacharparenleft}{\kern0pt}p\ {\isacharat}{\kern0pt}\ w{\isacharparenright}{\kern0pt}\ {\isasymLongrightarrow} \isanewline
\ {\isasymexists}r{\isachardot}{\kern0pt}\ u\ {\isacharequal}{\kern0pt}\ p\ {\isacharat}{\kern0pt}\ r\ {\isasymand}\ prefix\ r\ w{\isachardoublequoteclose}\isanewline
\isanewline
\isacommand{lemma}\isamarkupfalse%
\ example3\_sym{\kern0pt}{\isacharcolon}{\kern0pt}\ {\isachardoublequoteopen}suffix\ u\ {\isacharparenleft}{\kern0pt}p\ {\isacharat}{\kern0pt}\ w\ {\isacharat}{\kern0pt}\ q{\isacharparenright}{\kern0pt}\ {\isasymLongrightarrow}  \ length\ q\ {\isasymle}\ length\ u\ {\isasymLongrightarrow}\ length\ u\ {\isasymle}\ length\ {\isacharparenleft}{\kern0pt}w\ {\isacharat}{\kern0pt}\ q{\isacharparenright}{\kern0pt}\ {\isasymLongrightarrow}\isanewline
\ {\isasymexists}r{\isachardot}{\kern0pt}\ u\ {\isacharequal}{\kern0pt}\ r\ {\isacharat}{\kern0pt}\ q\ {\isasymand}\ suffix\ r\ w{\isachardoublequoteclose}\isanewline
\end{isaframe}

Applying the same strategy as for example2 fails, since trying to obtain example3\_sym from
\begin{isaframe}
example{\isadigit{3}}{\isacharbrackleft}{\kern0pt}of\ {\isachardoublequoteopen}rev\ u{\isachardoublequoteclose}\ {\isachardoublequoteopen}rev\ p{\isachardoublequoteclose}\ {\isachardoublequoteopen}rev\ w{\isachardoublequoteclose}\ {\isachardoublequoteopen}rev\ q{\isachardoublequoteclose}{\isacharcomma}{\kern0pt}\ unfolded\ symmetry{\isacharunderscore}{\kern0pt}rules{\isacharbrackright}{\kern0pt},
\end{isaframe}

\noindent where \isaterm{symmetry{\isacharunderscore}{\kern0pt}rules} is a list of appropriate symmetry rules, leaves us with

\begin{isaframe}
\ {\isachardoublequoteopen}suffix\ u\ {\isacharparenleft}{\kern0pt}{\isacharparenleft}{\kern0pt}q\ {\isacharat}{\kern0pt}\ w{\isacharparenright}{\kern0pt}\ {\isacharat}{\kern0pt}\ p{\isacharparenright}{\kern0pt}\ {\isasymLongrightarrow}
\ \ length\ p\ {\isasymle}\ length\ u\ {\isasymLongrightarrow} \isanewline\ length\ u\ {\isasymle}\ length\ {\isacharparenleft}{\kern0pt}w\ {\isacharat}{\kern0pt}\ p{\isacharparenright}{\kern0pt}\ {\isasymLongrightarrow}\ {\isasymexists}r{\isachardot}{\kern0pt}\ rev\ u\ {\isacharequal}{\kern0pt}\ rev\ p\ {\isacharat}{\kern0pt}\ r\ {\isasymand}\ prefix\ r\ {\isacharparenleft}{\kern0pt}rev\ w{\isacharparenright}{\kern0pt}{\isachardoublequoteclose},
\end{isaframe}

which is not yet in the form of example3\_sym.
This is not unexpected, the claim contains a bound variable \isaterm{r}. 
To finish the conversion, it suffices to realize that $(\exists x. \ P(x)) \leftrightarrow (\exists x. \ P(\rev{} x))$ holds.
Thus, we may replace the last two occurrences of \isaterm{r} with \isaterm{rev r}, and apply appropriate symmetric rules.

The next section addresses our realization of automatic production of symmetric rules, preceded by a discussion on the use of existing tools.

\section{Automated production of symmetrical claims} \label{sec:solution}

Before describing our solution to the automation of producing symmetrical claims, we discuss if and how might our task be achieved using tools for theorem reuse available in Isabelle/HOL.
We have not found any ready made tool in Isabelle/HOL that could achieve our objectives.
We shall briefly discuss two existing tools that achieve a similar task, namely reusing of a theory in a homomorphic setting.

The first tool are locales.
Locale is a mechanism for abstraction via interpretation and locale expressions \cite{Ballarin2010,Ballarin2014}.
We could see the ``by symmetry'' argument as two instantiations of the same claim: first in lists, and second in reversed lists.
It would require to prove all claims about lists in an abstract setting, and then apply it to lists and reversed lists.
The abstract setting would mean some kind of ``axiomatic theory of lists'', that is, of free monoids, as in \cite{HolubVeroff}.
While this may be the correct idea mathematically, we do not see how to naturally recreate it in Isabelle/HOL using locales.

The second tool is the infrastructure of transfer \cite{lifting_and_transfer,phd_kuncar}. 
Its main purpose is to transfer facts between two datatypes, e.g., from natural integers to integers, via user specified transfer rules.
Although, in principle, it should be possible to use this powerful tool, we encountered several problems using it and we did not find a way how to employ it for our purposes without producing undesired limitations.
For example, it is not clear how to specify whether in the case of transferring a fact containing $\alistlist$ the transfer rule should be applied to $\alist$ or $\alistlist = \blist$.

Since it seems from the above discussion that there is no direct way how to achieve the desired automation of the symmetry, we propose a ``lightweight'' solution which closely mimics the simple reproving of each individual claim ``on the fly'' as indicated by the examples in~\Cref{sec:context_and_examples}.
Our solution is very simple but at the same time it proves to be very practical and sufficiently versatile.

It is created as a single attribute called ``reversed''.
The symmetry rules are collected as a list of theorems called ``reversal\_rule'', i.e., a user can add and remove them any time.
By default, rules are required to eliminate reversal images, thus the reversal images are supposed to be on the left side of the equalities serving as rules.
For instance, the symmetry rule \Cref{p_to_s} is stored in this form
\begin{isaframe}
\ {\isachardoublequoteopen}suffix\ {\isacharparenleft}{\kern0pt}rev\ p{\isacharparenright}{\kern0pt}\ {\isacharparenleft}{\kern0pt}rev\ w{\isacharparenright}{\kern0pt}\ {\isacharequal}{\kern0pt}\ prefix\ p\ w{\isachardoublequoteclose}.
\end{isaframe}

The execution follows examples of \Cref{sec:ex2,sec:ex3}:
first, all schematic variables of type list of the fact being reversed are instantiated by their reversals.
Before the application of the symmetry rules, bound variables need to be treated.
Let us indicate this procedure on example3 of \Cref{sec:ex3} which contains one bound variable.
As indicated above, the idea is to use the equivalence $(\exists x. \ P(x)) \leftrightarrow (\exists x. \ P(\rev{} x))$.
We introduce a helper (private) definition and 2 claims as follows:
\begin{isaframe}
\isacommand{definition}\isamarkupfalse%
\ Ex{\isacharunderscore}{\kern0pt}rev{\isacharunderscore}{\kern0pt}wrap\ {\isacharcolon}{\kern0pt}{\isacharcolon}{\kern0pt}\ {\isachardoublequoteopen}{\isacharparenleft}{\kern0pt}{\isacharprime}{\kern0pt}a\ list\ {\isasymRightarrow}\ bool{\isacharparenright}{\kern0pt}\ {\isasymRightarrow}\ bool{\isachardoublequoteclose}\isanewline
\ \ \isakeyword{where}\ {\isachardoublequoteopen}Ex{\isacharunderscore}{\kern0pt}rev{\isacharunderscore}{\kern0pt}wrap\ P\ {\isacharequal}{\kern0pt}\ {\isacharparenleft}{\kern0pt}{\isasymexists}x{\isachardot}{\kern0pt}\ P\ {\isacharparenleft}{\kern0pt}rev\ x{\isacharparenright}{\kern0pt}{\isacharparenright}{\kern0pt}{\isachardoublequoteclose}\isanewline
\isanewline
\isacommand{lemma}\isamarkupfalse%
\ Ex{\isacharunderscore}{\kern0pt}rev{\isacharunderscore}{\kern0pt}wrapI{\isacharcolon}{\kern0pt}\ {\isachardoublequoteopen}{\isasymexists}x{\isachardot}{\kern0pt}\ P\ x\ {\isasymequiv}\ Ex{\isacharunderscore}{\kern0pt}rev{\isacharunderscore}{\kern0pt}wrap\ P{\isachardoublequoteclose}\isanewline
\isanewline
\isacommand{lemma}\isamarkupfalse%
\ Ex{\isacharunderscore}{\kern0pt}rev{\isacharunderscore}{\kern0pt}wrapE{\isacharcolon}{\kern0pt}\ {\isachardoublequoteopen}Ex{\isacharunderscore}{\kern0pt}rev{\isacharunderscore}{\kern0pt}wrap\ {\isacharparenleft}{\kern0pt}{\isasymlambda}x{\isachardot}{\kern0pt}\ P\ x{\isacharparenright}{\kern0pt}\ {\isasymequiv}\ {\isasymexists}x{\isachardot}{\kern0pt}\ P\ {\isacharparenleft}{\kern0pt}rev\ x{\isacharparenright}{\kern0pt}{\isachardoublequoteclose}\isanewline
\end{isaframe}

The application of these claims can be seen as
\begin{isaframe}
\ example{\isadigit{3}}{\isacharbrackleft}{\kern0pt}of\ {\isachardoublequoteopen}rev\ u{\isachardoublequoteclose}\ {\isachardoublequoteopen}rev\ p{\isachardoublequoteclose}\ {\isachardoublequoteopen}rev\ w{\isachardoublequoteclose}\ {\isachardoublequoteopen}rev\ q{\isachardoublequoteclose}{\isacharcomma}{\kern0pt}unfolded\ Ex{\isacharunderscore}{\kern0pt}rev{\isacharunderscore}{\kern0pt}wrapI{\isacharbrackright}{\kern0pt},
\end{isaframe}
which yields
\begin{isaframe}
\ {\isachardoublequoteopen}prefix\ {\isacharparenleft}{\kern0pt}rev\ u{\isacharparenright}{\kern0pt}\ {\isacharparenleft}{\kern0pt}rev\ p\ {\isacharat}{\kern0pt}\ rev\ w\ {\isacharat}{\kern0pt}\ rev\ q{\isacharparenright}{\kern0pt}\ {\isasymLongrightarrow} 
length\ {\isacharparenleft}{\kern0pt}rev\ p{\isacharparenright}{\kern0pt}\ {\isasymle}\ length\ {\isacharparenleft}{\kern0pt}rev\ u{\isacharparenright}{\kern0pt}\ {\isasymLongrightarrow}\isanewline
length\ {\isacharparenleft}{\kern0pt}rev\ u{\isacharparenright}{\kern0pt}\ {\isasymle}\ length\ {\isacharparenleft}{\kern0pt}rev\ p\ {\isacharat}{\kern0pt}\ rev\ w{\isacharparenright}{\kern0pt}\ {\isasymLongrightarrow} 
Ex{\isacharunderscore}{\kern0pt}rev{\isacharunderscore}{\kern0pt}wrap\ {\isacharparenleft}{\kern0pt}{\isasymlambda}r{\isachardot}{\kern0pt}\ rev\ u\ {\isacharequal}{\kern0pt}\ rev\ p\ {\isacharat}{\kern0pt}\ r\ {\isasymand}\ prefix\ r\ {\isacharparenleft}{\kern0pt}rev\ w{\isacharparenright}{\kern0pt}{\isacharparenright}{\kern0pt}{\isachardoublequoteclose}.
\end{isaframe}
The next step is
\begin{isaframe}
example{\isadigit{3}}{\isacharbrackleft}{\kern0pt}of\ {\isachardoublequoteopen}rev\ u{\isachardoublequoteclose}\ {\isachardoublequoteopen}rev\ p{\isachardoublequoteclose}\ {\isachardoublequoteopen}rev\ w{\isachardoublequoteclose}\ {\isachardoublequoteopen}rev\ q{\isachardoublequoteclose}{\isacharcomma}{\kern0pt}unfolded\ Ex{\isacharunderscore}{\kern0pt}rev{\isacharunderscore}{\kern0pt}wrapI{\isacharcomma}{\kern0pt}\ unfolded\ Ex{\isacharunderscore}{\kern0pt}rev{\isacharunderscore}{\kern0pt}wrapE{\isacharbrackright}{\kern0pt}
\end{isaframe}
resulting in
\begin{isaframe}
{\isachardoublequoteopen}prefix\ {\isacharparenleft}{\kern0pt}rev\ u{\isacharparenright}{\kern0pt}\ {\isacharparenleft}{\kern0pt}rev\ p\ {\isacharat}{\kern0pt}\ rev\ w\ {\isacharat}{\kern0pt}\ rev\ q{\isacharparenright}{\kern0pt}\ {\isasymLongrightarrow} 
length\ {\isacharparenleft}{\kern0pt}rev\ p{\isacharparenright}{\kern0pt}\ {\isasymle}\ length\ {\isacharparenleft}{\kern0pt}rev\ u{\isacharparenright}{\kern0pt}\ {\isasymLongrightarrow}\isanewline
length\ {\isacharparenleft}{\kern0pt}rev\ u{\isacharparenright}{\kern0pt}\ {\isasymle}\ length\ {\isacharparenleft}{\kern0pt}rev\ p\ {\isacharat}{\kern0pt}\ rev\ w{\isacharparenright}{\kern0pt}\ {\isasymLongrightarrow}\ {\isasymexists}r{\isachardot}{\kern0pt}\ rev\ u\ {\isacharequal}{\kern0pt}\ rev\ p\ {\isacharat}{\kern0pt}\ rev\ r\ {\isasymand}\ prefix\ {\isacharparenleft}{\kern0pt}rev\ r{\isacharparenright}{\kern0pt}\ {\isacharparenleft}{\kern0pt}rev\ w{\isacharparenright}{\kern0pt}{\isachardoublequoteclose}.
\end{isaframe}
Note that the name of the bound variable is preserved in this step.
It is due to \isaterm{{\isacharparenleft}{\kern0pt}{\isasymlambda}x{\isachardot}{\kern0pt}\ P\ x{\isacharparenright}{\kern0pt}} being present in 
\isaterm{Ex{\isacharunderscore}{\kern0pt}rev{\isacharunderscore}{\kern0pt}wrapE} rather than just \isaterm{P}.

This two step rewriting using the definition \isaterm{Ex\_rev\_wrap} in the intermediate step is to prevent an infinite loop of rewriting while trying to go directly from 
\isaterm{{\isasymexists}x{\isachardot}{\kern0pt}\ P\ x}
to
\isaterm{{\isasymexists}x{\isachardot}{\kern0pt}\ P\ {\isacharparenleft}{\kern0pt}rev\ x{\isacharparenright}}.

The last form is ready for the application of symmetry rules, and we almost obtain our goal, example3\_sym.
The remaining difference is the order of application of \isaterm{\isacharat}, i.e., the arrangement of parentheses.
The operation \isaterm{\isacharat} is associative and this final adjustment is left to be done manually, if desirable.

The implementation deals with other types of bound variables in a similar manner using a definition analogous to \isaterm{Ex\_rev\_wrap} and its two associated wrapping and unwrapping rules.
In a similar spirit, a special care for the constructors \Nil{} and \Cons{} is also part of the reversing process.
The described implementation is available at \cite{RevSymArch_gitlab}.

\section{Limits of the approach} \label{sec:limits}

To show the current limits, consider the following pair of symmetric claims:

\begin{isaframe}
\isacommand{lemma}\isamarkupfalse%
\ example{\isadigit{4}}{\isacharcolon}{\kern0pt}\ {\isachardoublequoteopen}prefix\ ps\ ws\ {\isasymLongrightarrow}\ prefix\ {\isacharparenleft}{\kern0pt}concat\ ps{\isacharparenright}{\kern0pt}\ {\isacharparenleft}{\kern0pt}concat\ ws{\isacharparenright}{\kern0pt}{\isachardoublequoteclose}
\isanewline \isanewline
\isacommand{lemma}\isamarkupfalse%
\ example{\isadigit{4}}{\isacharunderscore}{\kern0pt}sym{\isacharcolon}{\kern0pt}\ {\isachardoublequoteopen}suffix\ ps\ ws\ {\isasymLongrightarrow}\ suffix\ {\isacharparenleft}{\kern0pt}concat\ ps{\isacharparenright}{\kern0pt}\ {\isacharparenleft}{\kern0pt}concat\ ws{\isacharparenright}{\kern0pt}{\isachardoublequoteclose}\
\end{isaframe}

Applying the attribute reversed on example{\isadigit{4}} produces:

\begin{isaframe}
{\isachardoublequoteopen}suffix\ ps\ ws\ {\isasymLongrightarrow}\ prefix\ {\isacharparenleft}{\kern0pt}concat\ {\isacharparenleft}{\kern0pt}rev\ ps{\isacharparenright}{\kern0pt}{\isacharparenright}{\kern0pt}\ {\isacharparenleft}{\kern0pt}concat\ {\isacharparenleft}{\kern0pt}rev\ ws{\isacharparenright}{\kern0pt}{\isacharparenright}{\kern0pt}{\isachardoublequoteclose}
\end{isaframe}

The problem here is that we are dealing with variables of type \alistlist{}, representing factorizations or decomposition of words.
As they are of type \blist{}, the reversing happens only on this level, whereas to produce example4\_sym one would need the reversing to act on \alist{} = 
\isaterm{\isacharprime{}b}.
Namely, the additional required action of the symmetry on \isaterm{ps} and \isaterm{ws} is the application of \isaterm{map rev}.
The reason that this represents a current limit is that the choice of correct reversal rules for variables of type \alistlist{} becomes crucial and it is no more clear what are the correct reversal rules.

\section{Concluding remarks} \label{sec:conclusion}

Although the implemented attribute seems to be very simple, together with many delicately selected reversal rules it is very useful in our current project of formalization of Combinatorics on Words \cite{CoW_gitlab}.

As the attribute is a part of a living project, and the time period between the acceptance and publication of this article was noticeable, the obstacle exhibited in the previous section has been already surmounted in a way to suit the needs of the project.
However, the goal to properly deal with variables of any type remains.
In order to do that, our tentative model of the symmetry in question needs to be generalized and validated.

\subsubsection*{Acknowledgements}

The authors acknowledge support by the Czech Science Foundation grant GA\v CR 20-20621S.

\bibliographystyle{alpha}
\bibliography{../CICM2021/cicm2021}

\end{document}